\documentclass[12pt]{article}
\usepackage[]{amsmath,amssymb,amsfonts,latexsym,amsthm,enumerate,fullpage}
\usepackage{hyperref,cite}

\def\E{\mathbb{E}}
\def\F{\mathbb{F}}
\def\disc{\mathrm{disc}}

\def\eqdef{\stackrel{\textrm{def}}{=}}
\def\eps{\varepsilon}
\def\CC{\mathrm{CC}}
\def\CCdet{\CC^{\mathrm{det}}}
\def\rank{\mathrm{rank}}

\newtheorem{thm}{Theorem}[section]
\newtheorem*{thm*}{Theorem}
\newtheorem{claim}[thm]{Claim}
\newtheorem*{claim*}{Claim}
\newtheorem{lemma}[thm]{Lemma}
\newtheorem*{lemma*}{Lemma}

\newtheorem*{prop*}{Proposition}
\newtheorem{cor}[thm]{Corollary}
\newtheorem*{cor*}{Corollary}
\newtheorem{conj}[thm]{Conjecture}
\newtheorem*{conj*}{Conjecture}

\newtheorem*{dfn*}{Definition}
\theoremstyle{remark}

\newtheorem*{rmk*}{Remark}

\newtheorem*{rmks*}{Remarks}

\newcommand{\restate}[2]{\medskip \noindent{\bf #1 (restated).}{\sl #2}}

\title{Communication is bounded by root of rank}

\author{
Shachar Lovett\thanks{CSE department, UC San-Diego. e-mail: \texttt{slovett@cse.ucsd.edu}.}
}

\begin{document}

\maketitle

\begin{abstract}
We prove that any total boolean function of rank $r$ can be computed by a deterministic communication protocol of complexity $O(\sqrt{r} \cdot \log(r))$.
Equivalently, any graph whose adjacency matrix has rank $r$ has chromatic number at most $2^{O(\sqrt{r} \cdot \log(r))}$.
This gives a nearly quadratic improvement in the dependence on the rank over previous results.
\end{abstract}

\section{Introduction}

The \textit{log-rank conjecture} proposed by Lov\'{a}sz and Saks~\cite{LS88_Lat} suggests that for any boolean function $f:X \times Y \to \{-1,1\}$ its deterministic communication complexity $\CCdet(f)$ is polynomially related to the logarithm of the rank of the associated matrix. Validity of this conjecture is one of the fundamental open problems in communication complexity.
Very little progress has been made towards resolving it. The best upper bound, until recently, was
$$
\CCdet(f)\le \log(4/3) \cdot \rank(f),
$$
due to Kotlov~\cite{K97_Rank}. In terms of lower bounds, Kushilevitz (unpublished, cf.~\cite{NW94_On_Ra}) gave an example of a family of functions with $\CCdet(f) \ge (\log \rank(f))^{^{\log_36}}$.
Recently, a conditional improvement was made by Ben-Sasson, Ron-Zewi and the author~\cite{BLR12_An_Ad}, who showed that assuming a number-theoretic conjecture (the polynomial Freiman-Ruzsa conjecture), $\CCdet(f)\le O(\rank(f)/\log\rank(f))$. In this paper, we establish the following (unconditional) improved upper bound on the deterministic communication complexity.

\begin{thm}\label{thm:det}
Let $f:X \times Y \to \{-1,1\}$ be a boolean function with rank $r$. Then there exists a deterministic protocol computing $f$ which uses $O(\sqrt{r} \cdot \log r)$ bits of communication.
\end{thm}

The log-rank conjecture can be equivalently formulated as the relation between the rank of the adjacency matrix of a graph and its chromatic number. In this formulation, Theorem~\ref{thm:det} shows that any graph with adjacency matrix of rank $r$ has chromatic number at most $2^{O(\sqrt{r} \cdot \log r)}$.

%We also establish a version for randomized protocols and approximate rank. A randomized protocol is said to compute $f$ if for every input $x,y$, the protocol computes the correct value $f(x,y)$ with probability at least $2/3$. The approximate rank of $f$ is the minimal rank of a real $X \times Y$ matrix $M$ for which $2/3 \le M_{x,y} f(x,y) \le 1$ for all $x \in X, y \in Y$. Similar to the deterministic case, log of the approximate rank of $f$ is a lower bound on the randomized communication complexity of $f$. We prove the following upper bound.
%
%\begin{thm}\label{thm:rand}
%Let $f:X \times Y \to \{-1,1\}$ be a boolean function with approximate rank $r$. Then there exists a randomized protocol computing $f$ which uses $O(\sqrt{r})$ bits of communication.
%\end{thm}

\subsection{Proof overview}
The proof is based on analyzing the discrepancy of boolean functions. The discrepancy of a boolean function $f$ is given by
$$
\disc(f) = \min_{\mu} \max_R \left| \sum_{(x,y) \in R} f(x,y) \mu(x,y) \right|
$$
where $\mu$ ranges over all distributions over $X \times Y$ and $R$ ranges over all rectangles, e.g. $R=A \times B$ for $A \subset X, B \subset Y$.
Discrepancy is a well-studied property in the context of communication complexity lower bounds, see e.g.~\cite{Lokam09:book} for an excellent survey.
It is known that low-rank matrices have noticeable discrepancy~\cite{LMSS07_Com,LS09_Lea}: if $f$ has rank $r$ then
$$
\disc(f) \ge \frac{1}{8\sqrt{r}}.
$$
%A similar result holds if $f$ has approximate rank $r$, with a different constant.

Discrepancy can be used to prove upper bounds as well. Linial et al.~\cite{LMSS07_Com} showed that functions of discrepancy $\delta$ have randomized
(or quantum) protocols of complexity $O(1/\delta^2)$. Unfortunately, this does not give any improved bounds in general, as there is always a trivial
protocol using $r$ bits. We show that the combination of high discrepancy and low rank implies an improved bound. Our main new technical lemma shows
that if $f$ is a boolean function with discrepancy $\delta$, then there exist a large rectangle on which $f$ is nearly monochromatic. In the following,
we denote by $\E[f|R]$ the average value of $f$ on a rectangle $R$.

\begin{lemma}\label{lemma:amplification-intro}
Let $f:X \times Y \to \{-1,1\}$ be a function with $\disc(f)=\delta$. Then there exists a rectangle $R$ of size
$$
|R| \ge 2^{-O(\delta^{-1} \cdot \log(1/\eps))} |X \times Y|
$$
such that $\big| \E[f|R] \big| \ge 1-\eps$.
\end{lemma}

In fact, we prove a more general lemma which holds under general distributions. Now, if $f$ has low rank, we apply Lemma~\ref{lemma:amplification-intro} with $\eps=1/2r$ to deduce the existence of a large rectangle $R$ with $\big| \E[f|R] \big| \ge 1-1/2r$. Next, we apply the following claim from~\cite{GL13:lowrank_equiv}, which shows that low rank matrices which are nearly monochromatic contain large monochromatic rectangles.

\begin{claim}[\cite{GL13:lowrank_equiv}]\label{claim:mono-intro}
Let $f:X \times Y \to \{-1,1\}$ be a function with $\rank(f)=r$ and $\E[f|R] \ge 1-1/2r$. Then there exists a sub-rectangle $R' \subset R$ of size $|R'| \ge |R|/8$ such that $f$ is monochromatic on $R'$.
\end{claim}

Finally, we apply a theorem of Nisan and Wigderson~\cite{NW94_On_Ra}, who showed that in order to establish that low rank matrices have efficient deterministic protocols, it suffices to show that they have large monochromatic rectangles (which is what we just showed).

\begin{thm}[\cite{NW94_On_Ra}]\label{thm:NW}
Assume that for any function $f:X \times Y \to \{-1,1\}$ of $\rank(f)=r$ there exists a monochromatic rectangle of size $|R| \ge 2^{-c(r)} |X \times Y|$. Then any boolean function of rank $r$ is computable by a deterministic protocol of complexity $O(\log^2{r}+\sum_{i=0}^{\log{r}} c(r/2^i))$.
\end{thm}

As the proof in~\cite{NW94_On_Ra} is shown only for the special case related to the log-rank conjecture, we include a proof sketch of Theorem~\ref{thm:NW} for general function $c(r)$, in Section~\ref{sec:nw}. Theorem~\ref{thm:det} now follows by setting $c(r)=O(\sqrt{r} \cdot \log(r))$.

\subsection{Related works}
A recent work of Tsang et al~\cite{tsang2013fourier} established similar bounds to Theorem~\ref{thm:det} for the special case of functions of the form $f(x,y)=F(x \oplus y)$. Although the results
are similar, the techniques seem to be different. In particular, the main tool used in~\cite{tsang2013fourier} is Fourier analysis, while our results are based on discrepancy.
It would be interesting to understand if there are deeper connections between these techniques. Another recent work of Gavinsky and the author~\cite{GL13:lowrank_equiv} showed that in order
to prove the log-rank conjecture, it suffices to show that any low rank matrix has an efficient randomized protocol, a low information cost protocol, or an efficient zero-communication protocol.

\paragraph{Paper organization.} We give preliminary definitions in Section~\ref{sec:prelim}. We prove Lemma~\ref{lemma:amplification-intro} in Section~\ref{sec:amplification}. We prove Theorem~\ref{thm:det} in Section~\ref{sec:det}. We give a proof sketch of Theorem~\ref{thm:NW} in Section~\ref{sec:nw}. We discuss a conjecture related to matrix rigidity in Section~\ref{sec:rigidity}, and further open problems in Section~\ref{sec:summary}.

\section{Preliminaries}
\label{sec:prelim}

For standard definitions in communication complexity we refer the reader to~\cite{KN97_Com}. We give here only the basic definitions we would require.

Let $f:X \times Y \to \{-1,1\}$ be a total boolean function, where $X$ and $Y$ are finite sets. If $\mu$ is a distribution over $X \times Y$ then we denote by
$\E_{\mu}[f]=\sum_{x,y} \mu(x,y) f(x,y)$ the average of $f$ under $\mu$. A {\emph rectangle} is a set $R = A \times B$ for $A \subset X, B \subset Y$. We denote by $\E[f|R]$ the average of $f$ under the uniform distribution over $R$, and more generally by $\E_{\mu}[f|R]$ the average of $f$ under the conditional distribution of $\mu$ conditioned to be in $R$. A rectangle is {\emph monochromatic} if $f(x,y)=1$ for all $x,y \in R$ or $f(x,y)=-1$ for all $x,y \in R$.

The {\emph rank} of $f$ is the rank (over the reals) of its associated $X \times Y$ matrix.
%The \textit{approximate rank} of $f$ is the minimal rank of a real $X \times Y$ matrix $M$ such that $2/3 \le M_{x,y} f(x,y) \le 1$ for all $x \in X, y \in Y$.
The {\emph discrepancy} of $f$ with respect to a distribution $\mu$ on $X \times Y$ is the maximal bias achieved by a rectangle,
$$
\disc_{\mu}(f) \eqdef \max_{\textrm{rectangle }R} \left| \sum_{(x,y) \in R} \mu(x,y) f(x,y) \right|.
$$
The discrepancy of $f$ is the minimal discrepancy possible over all possible distributions $\mu$,
$$
\disc(f) \eqdef \min_{\mu} \disc_{\mu}(f).
$$

Note that discrepancy is an hereditary property. That is, if $R$ is a rectangle then the discrepancy of $f$ restricted to $R$ is at least the original discrepancy of $f$. Similarly, low rank is an hereditary property, as ranks of sub-matrices cannot exceed the rank of the original matrix. We will rely on the following theorem which lower bounds the discrepancy of functions with low rank.
\begin{thm}[\cite{LMSS07_Com,LS09_Lea}]\label{thm:disc}
Let $f:X \times Y \to \{-1,1\}$ be a function with rank $r$. Then $\disc(f) \ge 1/8\sqrt{r}$.
\end{thm}
%We note for the learned reader that while the proof in the literature is for low rank matrices, it generalizes immediately also to matrices of low approximate rank. The bound $\gamma_2 \le O(\sqrt{r})$ still holds for the approximating matrix (see e.g. Lemma 4.2 in~\cite{LMSS07_Com}), and by Grothendick inequality this implies a lower bound of $\Omega(1/\sqrt{r})$ on the discrepancy of $f$ (see e.g. Theorem 3.1 in~\cite{LS09_Lea}).

\section{An amplification lemma}
\label{sec:amplification}

Our main technical lemma is the following lemma, which shows that any boolean function with high discrepancy contains a large rectangle which is nearly monochromatic.
\begin{lemma}\label{lemma:amplification}
Let $f:X \times Y \to \{-1,1\}$ be a function with $\disc(f)=\delta$. Then for any $\eps>0$ and any distribution $\mu$ over $X \times Y$, there exists a rectangle $R$ with
$$
\mu(R) \ge 2^{-O(\delta^{-1} \cdot \log(1/\eps))}
$$
such that $\big| \E_{\mu}[f|R] \big| \ge 1-\eps$.
\end{lemma}

We note that Lemma~\ref{lemma:amplification-intro} from the introduction is a special case of Lemma~\ref{lemma:amplification} where $\mu$ is chosen to be the uniform distribution. Our original proof of Lemma~\ref{lemma:amplification} used an iterative amplification step. After giving a talk on this result in the Banff complexity workshop, Salil Vadhan suggested to us a simplified proof, which avoids the iterative step by applying Yao's mini-max principle. We present his proof below.

\begin{proof}
Let us assume without loss of generality that $\E_{\mu}[f] \ge 0$, otherwise apply the lemma to $-f$.
Let $\sigma$ be any distribution over $X \times Y$ such that $\E_{\sigma}[f]=0$. By assumption, there exists a rectangle $R_1$ such that
$$
\left|\sum_{(x,y) \in R_1} \sigma(x,y) f(x,y)\right| \ge \delta.
$$
Let $R_1 = A \times B$ and define $A' = X \setminus A, B' = Y \setminus B$. Consider the four rectangles
$$
R_1=A \times B, R_2 = A' \times B, R_3 = A \times B', R_4 = A' \times B'.
$$
As $\sum_{(x,y) \in X \times Y} \sigma(x,y) f(x,y) = \E_{\sigma}[f]=0$, there must exist a rectangle $R \in \{R_1,R_2,R_3,R_4\}$ such that
$$
\sum_{(x,y) \in R} \sigma(x,y) f(x,y) \ge \delta/3.
$$
As this holds for any distribution $\sigma$ for which $\E_{\sigma}[f]=0$, we can apply Yao's mini-max principle and deduce the following.
There exists a distribution $\rho$ over rectangles, such that, for any distribution $\sigma$ over $X \times Y$ for which $\E_{\sigma}[f]=0$, we have
$$
\E_{R \sim \rho}\left[ \sum_{(x,y) \in R} \sigma(x,y) f(x,y) \right] \ge \delta/3.
$$
Equivalently,
$$
\sum_{x \in X, y \in Y} \Pr_{R \sim \rho}[(x,y) \in R] \cdot \sigma(x,y) f(x,y) \ge \delta/3.
$$

Fix $(x_1,y_1) \in f^{-1}(1)$ and $(x_2,y_2) \in f^{-1}(-1)$. Let $\sigma$ be the distribution given by $\sigma(x_1,y_1)=\sigma(x_2,y_2)=1/2$. As $\E_{\sigma}[f]=0$ we have
$$
\Pr_{R \sim \rho}[(x_1,y_1) \in R] - \Pr_{R \sim \rho}[(x_2,y_2) \in R] \ge (2/3) \delta.
$$
Let $p$ be the \emph{minimal} probability that $(x_1,y_1) \in R$ over all $(x_1,y_1) \in f^{-1}(1)$, where $R$ is sampled according to $\rho$; and let $q$ be the \emph{maximal} probability that $(x_2,y_2) \in R$ over all $(x_2,y_2) \in f^{-1}(-1)$. We established that
$$
p - q \ge (2/3) \delta.
$$

Fix $t \ge 1$ and let $R_1,\ldots,R_t \sim \rho$ be chosen independently, and let $R^* = R_1 \cap \ldots \cap R_t$ be their intersection. We will show that for an appropriate choice of $t$, the rectangle $R^*$ satisfies the requirements of the lemma with positive probability (and hence such a rectangle exists). We will use the fact that for any $x \in X, y \in Y$,
$$
\Pr[(x,y) \in R^*] = \Pr_{R \sim \rho}[(x,y) \in R]^t.
$$
Consider the random variable
$$
T = \mu(R^*) - (1/\eps) \cdot \mu(R^* \cap f^{-1}(-1)).
$$
By linearity of expectation, we have
\begin{align*}
\E[T] &= \sum_{(x,y) \in f^{-1}(1)} \mu(x,y) \Pr[(x,y) \in R^*] - \sum_{(x,y) \in f^{-1}(-1)} \mu(x,y) ((1/\eps)-1) \Pr[(x,y) \in R^*] \\
& \ge \mu(f^{-1}(1)) \cdot p^t - \mu(f^{-1}(-1)) \cdot q^t / \eps \\
& \ge 1/2 \cdot (p^t - q^t / \eps),
\end{align*}
where we used our initial assumption that $\E_{\mu}[f]=\mu(f^{-1}(1))-\mu(f^{-1}(-1)) \ge 0$. We choose $t=O(p/\delta \cdot \log(1/\eps))$ so that
$$
q^t / p^t \le (1 - (2/3) \delta/p)^t \le \eps/2.
$$
For this choice of $t$, we have
$$
\E[T] \ge p^t / 4 = 2^{-O(\delta^{-1} \cdot \log(1/\eps))}.
$$
Let $R^*$ be a rectangle which achieves this average, that is
$$
\mu(R^*) - (1/\eps) \cdot \mu (R^* \cap f^{-1}(-1)) \ge 2^{-O(\delta^{-1} \cdot \log(1/\eps))}.
$$
In particular, we learn that both $\mu(R^*) \ge 2^{-O(\delta^{-1} \cdot \log(1/\eps))}$ (which satisfies the first requirement) and furthermore that $\mu(R^* \cap f^{-1}(-1)) \le \eps \cdot \mu(R^*)$, which implies that $\E_{\mu}[f|R^*] \ge 1-\eps$ (which satisfies the second requirement).
\end{proof}

\section{Deterministic protocols for low rank functions}
\label{sec:det}

We recall Theorem~\ref{thm:det} for the convenience of the reader.

\restate{Theorem~\ref{thm:det}}{
Let $f:X \times Y \to \{-1,1\}$ be a boolean function with rank $r$. Then there exists a deterministic protocol computing $f$ which uses $O(\sqrt{r} \cdot \log r)$ bits of communication.
}

We prove Theorem~\ref{thm:det} in the reminder of this section.
Let $f:X \times Y \to \{-1,1\}$ be a function of rank $r$. By Theorem~\ref{thm:disc} we have $\disc(f) \ge 1/8\sqrt{r}$. We apply Lemma~\ref{lemma:amplification} with $\eps=1/2r$ to derive the existence of a rectangle $R$ such that
$$
|R| \ge 2^{-O(\sqrt{r} \cdot \log(r))} \cdot |X \times Y|, \qquad \E[f|R] \ge 1-1/2r.
$$

Next, we apply a claim from~\cite{GL13:lowrank_equiv} which shows that nearly monochromatic rectangles in low rank matrices contain large monochromatic matrices.

\begin{claim}[\cite{GL13:lowrank_equiv}]\label{claim:mono}
Let $f:X \times Y \to \{-1,1\}$ be a function with $\rank(f)=r$ and $\E[f|R] \ge 1-1/2r$. Then there exists a rectangle $R' \subset R$ of size $|R'| \ge |R|/8$ such that $f$ is monochromatic on $R'$.
\end{claim}

For completeness, we include the proof.

\begin{proof}
Let $R=A \times B$.
Since $f$ is a sign matrix, the condition $\E[f|R] \ge 1-1/2r$ implies that $f(x,y)=-1$ for at most $1/4r$ fraction of the inputs in $R$.
Let $A' \subset A$ be the set of rows for which at most $1/2r$ fraction of the elements are $-1$,
$$
A'=\big\{x \in A: \left|\{y \in B: f(x,y)=-1\}\right| \le |B|/2r\big\}.
$$
By Markov inequality, $|A'| \ge |A|/2$. Let $x_1,\ldots,x_r \in A'$ be indices so that their rows span $A' \times B$. Let
$$
B'=\{y \in B: f(x_1,y)=\ldots=f(x_r,y)=1\}.
$$
Since each of the rows $x_1,\ldots,x_r$ contain at most $1/2r$ fraction of elements which are $-1$ we have $|B'| \ge |B|/2$. Now, this implies that all
rows in $A' \times B'$ are either the all one or all minus one. Choosing the largest half gives the required rectangle.
\end{proof}

Hence, we showed that any function $f:X \times Y \to \{-1,1\}$ of rank $r$ contains a monochromatic rectangle of size
$2^{-O(\sqrt{r} \cdot \log(r))} \cdot |X \times Y|$.
Applying Theorem~\ref{thm:NW} with $c(r) = O(\sqrt{r} \cdot \log(r))$, we conclude that any such function can be computed by a deterministic
 protocol which used $O(\sqrt{r} \cdot \log(r))$ bits of communication.

\subsection{Proof sketch of the Nisan-Wigderson theorem}
\label{sec:nw}

We recall Theorem~\ref{thm:NW} of Nisan and Wigderson~\cite{NW94_On_Ra} for the convenience of the reader.

\restate{Theorem~\ref{thm:NW}}{
Assume that for any function $f:X \times Y \to \{-1,1\}$ of $\rank(f)=r$ there exists a monochromatic rectangle of size $|R| \ge 2^{-c(r)} |X \times Y|$. Then any boolean function of rank $r$ is computable by a deterministic protocol of complexity $O(\log^2{r}+\sum_{i=0}^{\log{r}} c(r/2^i))$.
}

\begin{proof}
Let $f$ be a function of rank $r$, and consider the partition of its corresponding matrix as
$$
\left(
\begin{array}{cc}
R&S\\
P&Q
\end{array}
\right)
$$
As $R$ is monochromatic, $\rank(R)=1$. Hence, $\rank(S)+\rank(P) \le r+1$. Assume w.l.o.g that $\rank(S) \le r/2+1$ (otherwise, exchange the role of the rows and columns player). The row player sends one bit, indicating whether their input $x$ is in the top or bottom half of the matrix. If it is in the top half the rank decreases to $\le r/2+1$. If it is in the bottom half, the size of the matrix reduces to at most $(1-2^{-c(r)})|X \times Y|$. Iterating this process defines a protocol tree. We next bound the number of leaves of the protocol. By standard techniques, any protocol tree can be balanced so that the communication complexity is logarithmic in the number of leaves (cf.~\cite[Chapter 2, Lemma 2.8]{KN97_Com}).

Consider the protocol which stops once the rank drops to $r/2$. The protocol tree in this case has at most $O(2^{c(r)} \cdot \log(m))$ leaves, and hence can be simulated by a protocol sending only $O(c(r) + \log\log(m))$ bits. Note that since we can assume $f$ has no repeated rows or columns, $m \le 2^{2r}$ and hence $\log\log(m) \le \log(r)+1$. Next, consider the phase where the protocol continues until the rank drops to $r/4$. Again, this protocol can be simulated by $O(c(r/2)+\log(r))$ bits of communication. Summing over $r/2^i$ for $i=0,\ldots,\log(r)$ gives the bound.
\end{proof}

\section{A conjecture related to matrix rigidity}
\label{sec:rigidity}

The proof of Theorem~\ref{thm:det} relies on the matrix $f$ being boolean. However, we conjecture that it can be generalized to show that any low rank sparse matrix contains a large zero rectangle.

\begin{conj}\label{conj:sparse}
Let $M$ be an $n \times n$ real matrix with $\rank(M)=r$ and such that $M_{i,j} \ne 0$ for at most $\eps n^2$ entries. Then there exists $A,B \subset [n]$ such that
$$
M_{a,b}=0 \qquad \forall a \in A, b \in B
$$
such that $|A|,|B| \ge n \cdot \exp(-O(\sqrt{\eps r}))$.
\end{conj}

A related conjecture over $\F_2^n$, called the \emph{approximate duality conjecture}, was studied in~\cite{zewi2011affine,BLR12_An_Ad}, with relations to
two-source extractors and the log-rank conjecture. Here, we show that Conjecture~\ref{conj:sparse}, if true, would imply stronger bounds for matrix rigidity than currently known.

The bound in Conjecture~\ref{conj:sparse}, if true, is the best possible, as the following example shows. Let $M = N N^t$ where $N$ is an $n \times r$ matrix whose rows are all the $\{0,1\}^r$
vectors of hamming weight $\sqrt{r}/10$, and $n={r \choose \sqrt{r}/10} = r^{\Omega(\sqrt{r})}$. The matrix $M$ is $\eps=1/100$ sparse, as the probability that two uniformly chosen vectors
intersect is at most $1/100$. However, one can verify that the largest subsets $A, B \subset [n]$ such that $M_{a,b}=0$ for all $a \in A, b \in B$ correspond to choosing $A$ to be all
vectors whose support lies in the first half of the coordinates, and $B$ to be all vectors whose support lies in the last half of the coordinate. Furthermore, $|A|,|B| \le n \cdot \exp(-\Omega(\sqrt{r}))$.
The bound for general $\eps>0$ can be similarly obtained, by considering all vectors in $\{0,1\}^r$ of hamming weight $\sqrt{\eps r}$.

\paragraph{Matrix rigidity.}
A matrix $M$ is called $(r,s)$-rigid, if its rank cannot be made smaller than $r$ by changing at most $s$ entries in $M$. The problem of explicitly constructing rigid matrices was
introduced by Valiant~\cite{Valiant:rigidity} in the context of arithmetic circuits lower bounds, and was also studied by Razborov~\cite{Razborov89:rigid_cc} in
the context of separation of the analogs of PH and PSPACE in communication complexity. Despite much research, the best results to date are achieved by the so-called
"untouched minor" argument, which gives explicit matrices which are $(r,s)$-rigid with $s = \Omega\left(\frac{n^2}{r} \log\left(\frac{n}{r}\right)\right)$. See e.g. the excellent
survey of Lokam~\cite{lokam2009complexity} for details. We will prove the following corollary of Conjecture~\ref{conj:sparse}, which improves previous bounds by a logarithmic factor.

\begin{cor}
Assuming Conjecture~\ref{conj:sparse}, there exists an explicit $n \times n$ real matrix which is $(r,s)$-rigid for $s=\Omega\left(\frac{n^2}{r} \log^2\left(\frac{n}{r}\right)\right)$.
\end{cor}

\begin{proof}
Let $M$ be an $n \times n$ matrix of rank $r$, such that all $r \times r$ minors of $M$ have full rank. For example, such a matrix may be constructed as $M= N N^t$ where $N$ is an $n \times r$
matrix such that any $r$ rows of $N$ are linearly independent. Assume that $M$ is not $(r,s)$-rigid. Then, we can decompose
$$
M=L+S,\quad \rank(L)<r, \quad S \textrm{ is s-sparse}.
$$
Let $s=\eps n^2$. The matrix $S$ is both $s$-sparse and low rank, as $\rank(S) \le \rank(M)+\rank(L) < 2r$. Hence, by Conjecture~\ref{conj:sparse}, there exist $A,B \subset [n]$ of size
$|A|,|B| \ge n \cdot \exp(-O(\sqrt{\eps r}))$ such that $S_{a,b}=0$ for all $a \in A, b \in B$. Hence, $M_{a,b}=L_{a,b}$. If $|A|,|B| \ge r$, we must have that $\rank(L) \ge \rank(M)=r$. So,
$n \cdot \exp(-O(\sqrt{\eps r})) < r$ and the corollary follows by rearranging the terms.
\end{proof}

\section{Further research}
\label{sec:summary}

We provide a bound on the communication complexity that is near to linear in the discrepancy. This seem to be tight for our proof technique.
The dependence of the discrepancy on the rank, $\disc(f) \ge \Omega(1/\sqrt{\rank(f)})$, is tight in general, as can be seen for example
by taking $f$ to be the inner product function. However, it may be that further assuming that the rank of $f$ is much smaller than its size
might allow to prove better bounds. Another interesting direction is to combine our current approach with the additive combinatorics
approach of~\cite{BLR12_An_Ad}. Finally, we note that it may be possible to generalize the techniques developed here in order to relate
the approximate rank of a function and its randomized or quantum communication complexity.

\paragraph{Acknowledgements} I thank Dmitry Gavinsky, Pooya Hatami, Russell Impagliazzo and Adi Shraibman for helpful discussions, and Salil Vadhan 
for allowing to present his simplified proof of Lemma~\ref{lemma:amplification}.

\bibliographystyle{abbrv}
\bibliography{cc-root-rank}

\begin{thebibliography}{10}

\bibitem{BLR12_An_Ad}
E.~{B}en Sasson, S.~{L}ovett, and N.~Ron-{Z}ewi.
\newblock {An Additive Combinatorics Approach Relating Rank to Communication
  Complexity}.
\newblock {\em Proceedings of the 53rd Annual Symposium on Foundations of
  Computer Science}, pages 177--186, 2012.

\bibitem{zewi2011affine}
E.~{Ben-Sasson} and N.~Zewi.
\newblock From affine to two-source extractors via approximate duality.
\newblock In {\em Proceedings of the 43rd annual ACM symposium on Theory of
  computing}, pages 177--186. ACM, 2011.

\bibitem{GL13:lowrank_equiv}
D.~Gavinsky and S.~Lovett.
\newblock {En Route to the log-rank Conjecture: New Reductions and Equivalent
  Formulations}.
\newblock {\em Electronic Colloquium on Computational Complexity (ECCC)},
  20(80), 2013.

\bibitem{K97_Rank}
A.~Kotlov.
\newblock {Rank and Chromatic Number of a Graph}.
\newblock {\em Journal of Graph Theory 26(1)}, pages 1--8, 1997.

\bibitem{KN97_Com}
E.~Kushilevitz and N.~Nisan.
\newblock {Communication Complexity}.
\newblock {\em Cambridge University Press}, 1997.

\bibitem{LMSS07_Com}
N.~Linial, S.~Mendelson, G.~Schechtman, and A.~Schraibman.
\newblock {Complexity Measures of Sign Matrices}.
\newblock {\em Combinatorica 27(4)}, pages 439--463, 2007.

\bibitem{LS09_Lea}
N.~Linial and A.~Schraibman.
\newblock {Learning Complexity vs.\ Communication Complexity}.
\newblock {\em Combinatorics, Probability \& Computing 18(1-2)}, pages
  227--245, 2009.

\bibitem{lokam2009complexity}
S.~V. Lokam.
\newblock {\em Complexity Lower Bounds Using Linear Algebr}, volume~4.
\newblock Now Publishers Inc, 2009.

\bibitem{Lokam09:book}
S.~V. Lokam.
\newblock Complexity lower bounds using linear algebra.
\newblock {\em Found. Trends Theor. Comput. Sci.}, 4(1\&\#8211;2):1--155, Jan.
  2009.

\bibitem{LS88_Lat}
L.~Lov\'{a}sz and M.~Saks.
\newblock {Lattices, M\"obius Functions and Communication Complexity}.
\newblock {\em Annual Symposium on Foundations of Computer Science}, pages
  81--90, 1988.

\bibitem{NW94_On_Ra}
N.~Nisan and A.~Wigderson.
\newblock {On Rank vs.\ Communication Complexity}.
\newblock {\em Proceedings of the 35rd Annual Symposium on Foundations of
  Computer Science}, pages 831--836, 1994.

\bibitem{Razborov89:rigid_cc}
A.~Razborov.
\newblock On rigid matrices (in russian).
\newblock {\em Technical report, Steklov Mathematical Institute}, 1989.

\bibitem{tsang2013fourier}
H.~Y. Tsang, C.~H. Wong, N.~Xie, and S.~Zhang.
\newblock Fourier sparsity, spectral norm, and the log-rank conjecture.
\newblock {\em arXiv preprint arXiv:1304.1245}, 2013.

\bibitem{Valiant:rigidity}
L.~Valiant.
\newblock Graph-theoretic arguments in low-level complexity.
\newblock In J.~Gruska, editor, {\em Mathematical Foundations of Computer
  Science 1977}, volume~53 of {\em Lecture Notes in Computer Science}, pages
  162--176. Springer Berlin Heidelberg, 1977.

\end{thebibliography}

\end{document}